\documentclass[a4paper,onecolumn,accepted=2023-05-22]{quantumarticle}
\pdfoutput=1
\usepackage[utf8]{inputenc}
\usepackage[T1]{fontenc}
\usepackage[british]{babel}
\usepackage[unicode,colorlinks]{hyperref}
\usepackage{graphicx}
\usepackage{amsmath,amssymb,amsthm,bm}
\usepackage{xspace}
\usepackage{mathtools}
\usepackage{dsfont}
\usepackage{xcolor}
\usepackage{tikz,pgfplots}
\pgfplotsset{compat=1.17}

\usepackage{csquotes}
\usepackage[backend=bibtex,style=phys,biblabel=brackets,eprint=true,doi=false,maxnames=12,defernumbers=true,giveninits=true]{biblatex} 
\bibliography{biblio}

\DeclareMathOperator{\tr}{tr}

\let\originalleft\left
\let\originalright\right
\renewcommand{\left}{\mathopen{}\mathclose\bgroup\originalleft}
\renewcommand{\right}{\aftergroup\egroup\originalright}

\newcommand{\bra}[1]{\left\langle #1 \right|}
\newcommand{\ket}[1]{\left| #1 \right\rangle}

\newcommand{\ketbra}[2]{\left|#1\middle\rangle\middle\langle#2\right|}
\newcommand{\proj}[1]{\left|#1\middle\rangle\middle\langle#1\right|}
\newcommand{\Proj}[1]{|#1\rangle\langle#1|}

\newcommand{\norm}[1]{\left\|#1\right\|}

\newcommand{\de}[1]{\left(#1\right)}
\newcommand{\De}[1]{\left[#1\right]}
\newcommand{\DE}[1]{\left\{#1\right\}}

\newcommand{\id}{\mathds{1}}

\DeclareMathOperator{\dint}{d\!}

\newtheorem{theorem}{Theorem}
\newtheorem*{theorem*}{Theorem}

\mathtoolsset{centercolon}

\hypersetup{pdftitle={Quantum key distribution rates from semidefinite programming}}

\begin{document}
\title{Quantum key distribution rates from semidefinite programming}
\author{Mateus Araújo}\orcid{0000-0003-0155-354X}\affiliation{Institute for Quantum Optics and Quantum Information (IQOQI),
Austrian Academy of Sciences, Boltzmanngasse 3, 1090 Vienna, Austria}
\author{Marcus Huber}\affiliation{Institute for Quantum Optics and Quantum Information (IQOQI),
Austrian Academy of Sciences, Boltzmanngasse 3, 1090 Vienna, Austria}
\affiliation{Vienna Center for Quantum Science and Technology, Atominstitut, TU Wien, 1020 Vienna, Austria}
\author{Miguel Navascués}\affiliation{Institute for Quantum Optics and Quantum Information (IQOQI),
Austrian Academy of Sciences, Boltzmanngasse 3, 1090 Vienna, Austria}
\author{Matej Pivoluska}
\affiliation{Vienna Center for Quantum Science and Technology, Atominstitut, TU Wien, 1020 Vienna, Austria}
\affiliation{Institute of Computer Science, Masaryk University, 602 00 Brno, Czech Republic}
\affiliation{Institute of Physics, Slovak Academy of Sciences, 845 11 Bratislava, Slovakia}
\author{Armin Tavakoli}\affiliation{Institute for Quantum Optics and Quantum Information (IQOQI),
Austrian Academy of Sciences, Boltzmanngasse 3, 1090 Vienna, Austria}
\affiliation{Vienna Center for Quantum Science and Technology, Atominstitut, TU Wien, 1020 Vienna, Austria}
\date{22nd May 2023}

\begin{abstract}
Computing the key rate in quantum key distribution (QKD) protocols is a long standing challenge. Analytical methods are limited to a handful of protocols with highly symmetric measurement bases. Numerical methods can handle arbitrary measurement bases, but either use the min-entropy, which gives a loose lower bound to the von Neumann entropy, or rely on cumbersome dedicated algorithms. Based on a recently discovered semidefinite programming (SDP) hierarchy converging to the conditional von Neumann entropy, used for computing the asymptotic key rates in the device independent case, we introduce an SDP hierarchy that converges to the asymptotic secret key rate in the case of characterised devices. The resulting algorithm is efficient, easy to implement and easy to use. We illustrate its performance by recovering known bounds on the key rate and extending high-dimensional QKD protocols to previously intractable cases. We also use it to reanalyse experimental data to demonstrate how higher key rates can be achieved when the full statistics are taken into account.
\end{abstract}
\maketitle

\section{Introduction}

In the long history of QKD (see \cite{gisin2002, scarani2009,xu2020,pirandola2020} for reviews), secret key rates have mostly been calculated analytically. Consequently, we only know their values for protocols with highly symmetric measurement bases, such as BB84 \cite{bennett1984}, the six-state protocol \cite{bruss1998}, or their generalisations to mutually unbiased bases (MUBs) in higher dimensions \cite{cerf2002,sheridan2010}. In order to tackle arbitrary measurement bases, for a long time the only approach available was to use the min-entropy to lower bound the von Neumann entropy \cite{koenig2009}. This trick has been applied successfully to device independent QKD \cite{bancal2014,nieto-silleras2014} and QKD with characterised devices \cite{doda2021}. It leads to a simple SDP, at the cost of delivering suboptimal key rates. Numerical methods that go beyond the min-entropy approximation have been developed to tackle this issue \cite{wang2019,tan2021}; unfortunately, none of them output optimal key rates. In Refs.~\cite{winick2018,hu2022} numerical methods that do provide optimal key rates are proposed. We'll compare them with our method in the Conclusion.

Recently, Brown et al. \cite{brown2021} presented a new approach for the accurate computation of asymptotic secret key rates in device independent QKD. It is based on a converging SDP hierarchy of lower bounds on the conditional von Neumann entropy. This hierarchy relies, in turn, on the NPA hierarchy \cite{navascues2008} for characterising quantum correlations. As a consequence, the Brown et al. hierarchy has a very high complexity, being capable of handling only cases with few measurement settings.

In this work we adapt the SDP hierarchy proposed in Ref.~\cite{brown2021} to the case of QKD with characterised devices, enabling likewise the computation of a sequence of lower bounds converging to the asymptotic key rate. In contrast with the device independent case, we don't need to use the NPA hierarchy, which makes the method much more efficient. Remarkably, the complexity of this new method is essentially independent of the number of measurement settings. Instead, the most relevant parameter is the quantum state dimension, and we can handle cases up to local dimension $8$.

We illustrate the power of the technique on two families of protocols. First, for protocols that use $d+1$ mutually-unbiased bases (MUB protocols), we recover the analytical key rate in the known cases, and extend the protocol to use the full experimental data in any dimension (including dimensions that are not primes or prime powers, in which we use $d+1$ bases that are only approximately unbiased). Second, we propose and analyse a new protocol that is tailored to situations where one can produce high-dimensional entanglement, but can only measure superpositions of consecutive basis elements.

In order to handle real experimental data, we show how to modify the algorithm to minimise the conditional entropy over a confidence region in parameter space, and propose a new Monte Carlo method to compute such a confidence region via Bayesian parameter estimation. We illustrate the method by reanalysing data from two experiments.

The paper is organised as follows. In Section \ref{sec:brown} we recap necessary results from Ref.~\cite{brown2021}. In Section \ref{sec:mainsdp} we present and prove the main result, the SDP hierarchy for bounding QKD rates. In Section \ref{sec:numerics} we apply it to recover known bounds on the key rate and compute new ones. In Section \ref{sec:experimental} we show how to modify the SDP to handle experimental data, and use that to bound key rates from two previous experiments.


\section{Variational lower bounds for the conditional entropy}\label{sec:brown}

In this section, we review necessary results from Ref.~\cite{brown2021} and establish the notation to be used from now on.

In quantum key distribution, the asymptotic key rate is lowerbounded by the Devetak-Winter rate\footnote{This expression gives the rate per key round. In order to get the rate per round one needs in addition to multiply it by the sifting factor, or notice that the sifting factor can be made arbitrarily close to one in the asymptotic limit by choosing the key basis with probability arbitrarily close to one \cite{Lo2005}.} \cite{devetak2005}
\begin{equation}
K \ge H(A|E)-H(A|B),
\end{equation}
where $H(\cdot|\cdot)$ is the quantum conditional entropy. $H(A|B)$ only depends on the statistics measured by Alice and Bob in the key basis, and as such is straightforward to compute. The interesting problem is $H(A|E)$, as we have to minimise it over all states $\rho_{ABE}$ share by Alice, Bob, and Eve compatible with the statistics measured by Alice and Bob. To do that, we express the conditional entropy in terms of the (unnormalized) relative entropy:
\begin{equation}
H(A|E) = -D(\rho_{\tilde{A}E}||\id_{A} \otimes \rho_E),
\end{equation}
where $D(\rho||\sigma) = \tr(\rho(\log_2\rho-\log_2\sigma))$ and $\rho_{\tilde{A}E}$ is the classical-quantum state given by $\sum_{a} \proj{a} \otimes \rho_E(a)$, where $\rho_E(a) = \tr_{AB} \De{(A^a_0 \otimes \id_{BE}) \rho_{ABE}}$ and $\{A^a_0\}_{a=0}^{n-1}$ is the basis used by Alice to generate the secret key.

As shown in Ref.~\cite{brown2021}, one can get a convergent sequence of SDPs for the relative entropy by using a Gauss-Radau quadrature \cite{golub1973}. The Gauss-Radau quadrature is defined for every positive integer $m$ as a vector of $m$ coordinates $\bm{t}$ and a vector of $m$ weights $\bm{w}$ such that for all polynomials $g$ of degree $2m-2$ it holds that
\begin{equation}
\int_0^1\dint t\,g(t) = \sum_{i=1}^m w_i g(t_i).
\end{equation}
In addition we have that $w_i > 0$, $\sum_i w_i = 1$, $w_m = 1/m^2$, $t_i \in (0,1]$, and $t_m = 1$. A simple algorithm for computing $\bm{w}$ and $\bm{t}$ is described in \cite{golub1973}. 

As shown in Ref.~\cite{brown2021}, applying the Gauss-Radau quadrature to an integral representation of the logarithm yields a variational upper bound for the quantum relative entropy, valid when $\operatorname{supp}(\rho)\cap \operatorname{ker}(\sigma) = \varnothing$:
\begin{equation}
D(\rho||\sigma) \le -\sum_{i=1}^m \frac{w_i}{t_i \log 2} \inf_{Z_i} \de{1+ \tr\De{\rho(Z_i+Z_i^\dagger + (1-t_i) Z_i^\dagger Z_i)} + t_i \tr\de{\sigma Z_i Z_i^\dagger}},
\end{equation}
where the $Z_i$ are arbitrary complex matrices. This translates into a variational lower bound on the quantum conditional entropy for a fixed state $\rho_{ABE}$:
\begin{multline}
H(A|E)_{\rho_{ABE}} \ge \\ c_m + \sum_{i=1}^m \frac{w_i}{t_i \log 2} \inf_{\{Z_i^a\}_a} \sum_{a=0}^{n-1} \tr\De{\rho_{ABE}\de{A^a_0 \otimes \id_B \otimes (Z_i^a + {Z_i^a}^\dagger + (1-t_i){Z_i^a}^\dagger Z_i^a) + t_i \id_{AB} \otimes Z_i^a{Z_i^a}^\dagger}},
\end{multline}
where $c_m = \sum_{i=1}^m \frac{w_i}{t_i \log 2}$.

In the case of device independent QKD, one needs to minimise $H(A|E)$ over all states and measurement bases of Alice and Bob compatible with the measured statistics; here, though, we are interested in QKD with characterised devices. We take the measurement bases of Alice and Bob to be fixed, and do the minimisation over the ${\rho_{ABE}}$ that are compatible with the measured statistics. That is, we want to solve the following problem:
\begin{align}\label{eq:qkd_problem}
\inf_{\rho_{ABE},\{Z_i^a\}_{a,i}} c_m + \sum_{i=1}^m\sum_{a=0}^{n-1} \frac{w_i}{t_i \log 2} \tr & \De{\rho_{ABE}\de{A^a_0 \otimes \id_B \otimes (Z_i^a + {Z_i^a}^\dagger + (1-t_i){Z_i^a}^\dagger Z_i^a) + t_i \id_{AB} \otimes Z_i^a{Z_i^a}^\dagger}} \nonumber \\
&\text{s.t.}\quad \rho_{ABE} \ge 0,\quad \tr(\rho_{ABE}) = 1 \\
&\ \forall k\quad\tr\De{\rho_{ABE} (E_k \otimes \id_E)} = f_k, \nonumber 
\end{align}
for some POVM elements $E_k$ and probabilities $f_k$. In the next section we'll show how to do this with a SDP.

\section{SDP}\label{sec:mainsdp}

To turn problem \eqref{eq:qkd_problem} into a SDP, two ideas are involved: first is the familiar technique of absorbing the variables $Z_i^a$ into $\rho_{ABE}$ by defining new variables $\zeta_i^a = \tr_E \De{\rho_{ABE}(\id_{AB} \otimes {Z_i^a}^T ) }$. This still leaves the objective nonlinear, as it contains the products ${\zeta_i^a}^\dagger \zeta_i^a$ and $\zeta_i^a {\zeta_i^a}^\dagger$. The usual way to handle those is by building a moment matrix, but the difficulty here is that these must be $d^2 \times d^2$ matrices, and the usual moment matrix technique cannot handle these dimensional constraints. What we do instead is to build a \emph{block} moment matrix, representing each element of ${\zeta_i^a}^\dagger \zeta_i^a$ and $\zeta_i^a {\zeta_i^a}^\dagger$ by a moment, as done in Section 3 of Ref.~\cite{navascues2014}.

More formally, we have
\begin{theorem}
The problem \eqref{eq:qkd_problem} defined in the previous section is equivalent to the following SDP:

\begin{equation}\label{eq:qkd_sdp}
\begin{gathered}
\min_{\sigma,\{\zeta^a_i,\eta^a_i,\theta^a_i\}_{a,i}} c_m + \sum_{i=1}^m\sum_{a=0}^{n-1} \frac{w_i}{t_i \log 2} \tr\De{(A^a_0 \otimes \id_B)\de{\zeta^a_i + {\zeta^a_i}^\dagger + (1-t_i)\eta^a_i} + t_i\theta^a_i} \\
\text{s.t.}\quad \tr\de{\sigma} = 1,\quad \forall k\  \tr(E_k \sigma) = f_k\\
\forall a,i \quad \Gamma^1_{a,i} := \begin{pmatrix} 
\sigma & \zeta^a_i \\
{\zeta^a_i}^\dagger & \eta^a_i
\end{pmatrix}  \ge 0,\quad \Gamma^2_{a,i} := \begin{pmatrix} 
\sigma & {\zeta^a_i}^\dagger \\
\zeta^a_i & \theta^a_i
\end{pmatrix} \ge 0.
\end{gathered}
\end{equation}
\end{theorem}
\begin{proof}
First notice that if we define\footnote{We could also define it without the transpose, which would make it a completely copositive map, but this would result in a less convenient SDP.}
\begin{equation}\label{eq:cpmap}
\Xi(M) := \tr_E \De{\rho_{ABE}(\id_{AB} \otimes M_E^T) }
\end{equation}
then $\Xi(M^\dagger) = \Xi(M)^\dagger$ and 
\begin{equation}
\tr\De{\rho_{ABE}\de{K_{AB} \otimes M_E^T}} = \tr\de{K_{AB} \Xi(M_E)},
\end{equation}
so if we define
\begin{subequations}\label{eq:xiequations}
\begin{gather}
\sigma := \Xi(\id) \\
\zeta^a_i := \Xi(Z^a_i) \\
\eta^a_i := \Xi({Z^a_i}^\dagger Z^a_i) \\
\theta^a_i := \Xi(Z^a_i{Z^a_i}^\dagger )
\end{gather}
\end{subequations}
then the objectives and equality constraints of \eqref{eq:qkd_problem} and \eqref{eq:qkd_sdp} match, modulo a complex conjugation on the $Z^a_i$ that we can freely do, since they are unconstrained.

It remains to show that the existence of a state $\rho_{ABE}$ and complex matrices $Z_i^a$ is equivalent to the existence of positive semidefinite matrices $\Gamma^1_{a,i}$ and $\Gamma^2_{a,i}$ related via Equations \eqref{eq:xiequations}. To see the forward direction, first note that the map $\Xi$ is completely positive: its Choi matrix is given by
\begin{align}
\mathfrak{C}(\Xi) &= \sum_{ij} \ket{i}\bra{j} \otimes \Xi(\ket{i}\bra{j}) \\
				  &= \sum_{ij} \ket{i}\bra{j} \otimes \De{(\id \otimes \bra{i}) \rho (\id \otimes \ket{j})} \\
   				  &= \bigg(\sum_{i} \ket{i}\otimes \id \otimes \bra{i}\bigg) \rho \bigg(\sum_{j}\ket{j} \otimes \id \otimes \bra{j}\bigg)^\dagger,
\end{align}
which is manifestly positive semidefinite.

Let then
\begin{equation}
\gamma^g_{a,i} := \sum_{k,l=0}^1 \ket{k}\bra{l} \otimes {s^g_k}^\dagger s^g_l
\end{equation}
for $s^1_0 = s^2_0 = \id$ and $s^1_1 = {s^2_1}^\dagger = Z^a_i$. The matrices $\gamma^g_{a,i}$ are positive semidefinite, as
\begin{equation}
\gamma^g_{a,i} = \de{\sum_{k=0}^1 \bra{k} \otimes s^g_k}^\dagger\de{\sum_{l=0}^1 \bra{l} \otimes s^g_l},
\end{equation}
and so are $\Gamma^1_{a,i}$ and $\Gamma^2_{a,i}$, as
\begin{equation}
\Gamma^g_{a,i} = (\id \otimes \Xi)(\gamma^g_{a,i})
\end{equation}
and $\Xi$ is completely positive. This completes the proof of the forward direction.

To see the converse direction, let the eigendecomposition of $\sigma$ be $\sum_{i=0}^{d^2-1} \lambda_i \ket{v_i}\bra{v_i}$, and consider its purification $\ket{\psi} := \sum_{i=0}^{d^2-1} \sqrt{\lambda_i} \ket{v_i}\ket{i}_E$. In the following, we assume that Eve's Hilbert space is $\mathcal{H}_E:=\operatorname{span}\{\ket{i}\}_{i=0}^{2d^2-1}$, and regard the state shared by Alice, Bob and Eve as $\rho_{ABE} = \ket{\psi}\bra{\psi}$. Eve's reduced density matrix has therefore support on the space spanned by just the first $d^2$ basis states of $\mathcal{H}_E$. The ``extra space'' given by $\operatorname{span}\{\ket{i}\}_{i=d^2-1}^{2d^2-1}$ will be needed by the end of the proof.

Now, define the matrix
\begin{equation}
W := \sum_{i=0}^{d^2-1} \sqrt{\lambda_i}\ket{v_i}\bra{i}.
\end{equation}
With the state $\rho_{ABE}$ defined above, the identity $\Xi(M) = WMW^\dagger$ can be seen to hold.

Define then 
\begin{equation}\label{eq:pos1}
\begin{pmatrix} 
\tilde\id & R_{\zeta^a_i} \\
R_{\zeta^a_i}^\dagger & R_{\eta^a_i}
\end{pmatrix} := \begin{pmatrix} W^+ & 0\\ 0 & W^+ \end{pmatrix}\Gamma^1_{a,i}\begin{pmatrix} {W^+}^\dagger & 0\\ 0 & {W^+}^\dagger \end{pmatrix},
\end{equation}
and
\begin{equation}\label{eq:pos2}
\begin{pmatrix} 
\tilde\id & {R^\dagger_{\zeta^a_i}} \\
R_{\zeta^a_i} & R_{\theta^a_i}
\end{pmatrix} := \begin{pmatrix} W^+ & 0\\ 0 & W^+ \end{pmatrix}\Gamma^2_{a,i}\begin{pmatrix} {W^+}^\dagger & 0\\ 0 & {W^+}^\dagger \end{pmatrix}.
\end{equation} Here $W^+$ denotes the pseudo-inverse of $W$; and $\tilde\id$, the projector corresponding to the support of $\sigma$.

The matrices on the left-hand side of Equations \eqref{eq:pos1} and \eqref{eq:pos2} are positive semidefinite, which implies that $R_{\eta^a_i} - R_{\zeta^a_i}^\dagger R_{\zeta^a_i} \ge 0$ and that $R_{\theta^a_i} - R_{\zeta^a_i}R_{\zeta^a_i}^\dagger \ge 0$. We are ready to build Eve's operators $Z^a_i\in B(\mathcal{H}_E)$:
\begin{equation}
Z^a_i := \begin{pmatrix} R_{\zeta^a_i} & \sqrt{R_{\theta^a_i} - R_{\zeta^a_i}R_{\zeta^a_i}^\dagger} \\
						\sqrt{R_{\eta^a_i} - R_{\zeta^a_i}^\dagger R_{\zeta^a_i}} & 0 \end{pmatrix},
\end{equation}
where $\sqrt{\cdot}$ denotes the unique positive semidefinite square root. In the expression above, the matrix block decomposition reflects the division of Eve's Hilbert space $\mathcal{H}_E$ into the orthogonal subspaces $\operatorname{span}\{\ket{i}\}_{i=0}^{d^2-1}$ and $\operatorname{span}\{\ket{i}\}_{i=d^2}^{2d^2-1}$. With the definition above, it is easy to see that Equations \eqref{eq:xiequations} are satisfied. This completes the proof of the converse direction.
\end{proof}

Unlike the case for device independent QKD, for fixed $m$ this is not a hierarchy, but a single SDP. We thus have a converging hierarchy for the conditional entropy depending on a single parameter $m$. 

In order to obtain a reliable solution for an SDP numerically, it is crucial that it satisfies strict feasibility, that is, that there exists a feasible solution for which all eigenvalues in the positive semidefiniteness constraints are strictly positive. For SDP \eqref{eq:qkd_sdp} this is the case iff there exists a full rank state $\sigma$ compatible with the measurement results. If this is not the case then one must reformulate the SDP with facial reduction to make it strictly feasible \cite{drusvyatskiy2017}. This is of little relevance in practice, as in any real experiment the state will be full rank. For completeness, however, we give a proof of this characterisation of strict feasibility and show how to perform facial reduction in Appendix \ref{sec:facialreduction}.

The SDP \eqref{eq:qkd_sdp} takes a long time to run in higher dimensions, so it is important to look for techniques to optimise it. We have explored the use of symmetrisation \cite{gatermann2004}: if the objective of an SDP is invariant under some transformation applied to its variables and this transformation doesn't affect whether the constraints are satisfied, then this transformation is a symmetry of the SDP and can be used to eliminate redundant variables, which can dramatically improve the running time.

A simple and powerful symmetrisation applies when $A^a_0$ and $E_k$ are real matrices for all $a,k$: in this case we can choose all our variables to be real, which has a dramatic impact on performance, mainly due to the poor support current solvers have for complex variables\footnote{To the best of our knowledge, the only solvers that support complex numbers natively are SeDuMi 1.3.7 \cite{sturm1999} and Hypatia 0.7.2 \cite{Coey2021}.}. Another useful symmetrisation applies when the $E_k$ satisfy a certain permutation symmetry, which allows us to reduce the number of variables $\zeta^a_i,\eta^a_i,\theta^a_i$ by at most a factor of $d$, which again can dramatically improve performance. This is proven and explored in detail in Appendix \ref{sec:symmetrisation}. The focus of our paper, however, is using SDP \eqref{eq:qkd_sdp} to calculate key rates using the full experimental data available, in which case the symmetrisation seldom applies.

\section{Numerical results}\label{sec:numerics}

In order to illustrate the performance of our technique, we ran the SDP \eqref{eq:qkd_sdp} for various QKD protocols. We solved the SDP with MOSEK \cite{mosek} via YALMIP\footnote{It is necessary to use the current development release.} \cite{yalmip} on MATLAB\footnote{The code also works on Octave, provided we change the solver to SeDuMi \cite{sturm1999}.}. In order to model the effect of noise, we assume that Alice and Bob share an isotropic state:
\begin{equation}
\rho_\text{iso}(v) = v \proj{\phi^+} + (1-v) \id/d^2,
\end{equation}
where $\ket{\phi^+} = \frac1{\sqrt{d}}\sum_{i=0}^{d-1}\ket{ii}$ and $v \in [0,1]$ is the visibility.

\subsection{MUB protocol}\label{sec:mubprotocol}

The first protocol we consider is the $d+1$ MUBs protocol from Ref.~\cite{sheridan2010}, where Alice and Bob measure a complete set of MUBs for some prime $d$, with Bob's MUBs being the transpose of Alice's. They estimate the probabilities of their outcomes being equal or differing by some integer, and use only this data to compute the key rate. 

Both these limitations, that the dimension must be prime and the full data not being used, came from the proof technique and do not apply here. Therefore we extend the protocol to use all available data, and furthermore allow the dimension $d$ to be any integer. When $d$ is equal to a prime power we one can use the well-known Wootters-Fields construction \cite{wootters1989} to generate the MUBs. For other dimensions it is widely believed that no set of $d+1$ MUBs exists \cite{bengtsson2007}, so we generate numerically bases that are only approximately unbiased via a simple gradient descent. A set of $n$ approximate MUBs is obtained from numerically minimising \cite{bengtsson2007b}
\begin{equation}
    \min_{\{U^{(r)}\}_r}\frac1{(d-1)\binom{n}{2}}\sum_{k<l}  \sum_{i,j=0}^{d-1} \left( \left|({U^{(k)}}^\dagger U^{(l)})_{i,j}\right|^2 -\frac{1}{d}\ \right)^2,
\end{equation}
where $\{U^{(r)}\}_{r=1}^n$ is a collection of unitary matrices. Note that the value is zero for a set of MUBs.

The key basis is the computational basis, that is, $A^a_0 = \proj{a}$, and the POVM elements used in the SDP are given by
\begin{equation}
E_{k,i,j} := \Pi^i_k \otimes {\Pi^j_k}^T,
\end{equation}
where $\Pi^i_k$ is the projector onto the $i$th vector of the $k$th MUB, with $k$ going from 0 to $d$ and $i,j$ going from 0 to $d-1$.

In order to obtain an analytical formula to compare against our numerical results, note that when we actually have $d+1$ MUBs, that is, when $d$ is a prime or a prime power, these measurements are tomographically complete for the isotropic state\footnote{More precisely, if Alice and Bob measure the probability of getting the same outcome in all $d+1$ MUBs, and all probabilities turn out to be equal to $v + (1-v)/d$, then the only quantum state compatible with these results is the isotropic state with visibility $v$. Note that this implies that using the full data does not make any difference for the isotropic state when $d$ is a prime or a prime power. Nevertheless it does make a difference for general states in general dimensions, so it is still useful to define the protocol to use the full data.} \cite{bavaresco2018}, so the key rate must be equal to the tomographic rate. That is the rate one obtains when Alice and Bob make a tomographically complete set of measurements, and as such is the best possible rate for a given quantum state. When our $d+1$ bases are not exactly mutually unbiased this no longer holds, so for other dimensions the tomographic rate is only an upper bound for the key rate.

For the isotropic state the tomographic rate is given by \cite{liang2003,zhan2020}
\begin{equation}\label{eq:tomographic}
K_\text{iso}(v,d) = \log_2(d) -(1-1/d^2)(1-v) \log_2(d^2-1)-h(v+(1-v)/d^2),
\end{equation}
where $h(\cdot)$ is the binary entropy. Note that Equation \eqref{eq:tomographic} coincides with the key rate computed in Ref.~\cite{sheridan2010} for the isotropic state when $d$ is prime.

We did the calculation for $d=2,4,6,8$, setting $m=8$. As we can see in Figure \ref{fig:mub}, our technique closely matches the exact key rate for $d=2,4,8$, whereas for $d=6$ it is slightly below the upper bound. This indicates that the problem of existence of MUBs, although mathematically fascinating, is irrelevant for the practical implementation of this QKD protocol.

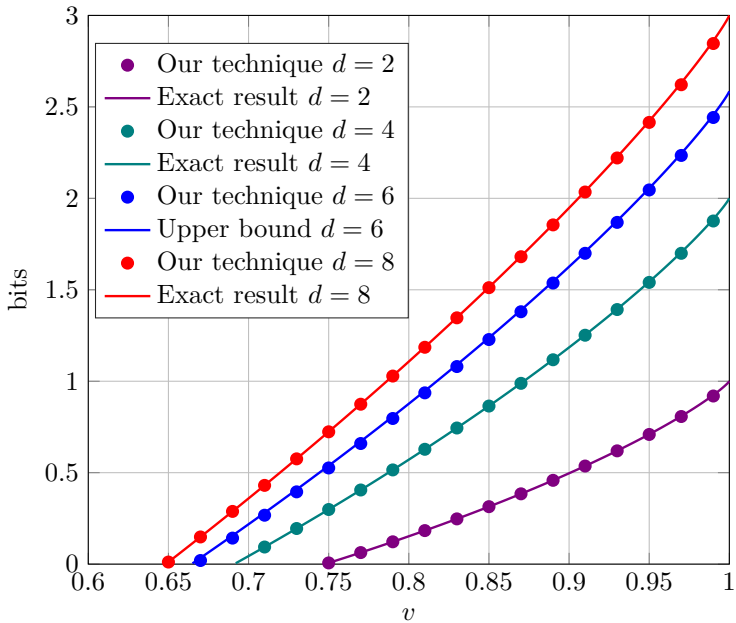
\begin{figure}[htb]
	\centering
	\begin{tikzpicture}
		
		\begin{axis}[%
			scale only axis,
			xmin=0.6,
			xmax=1,
			ymin=0,
			ymax=3,
			grid=major,
			xlabel={$v$},
			ylabel={bits},
			xtick={0.60, 0.65, 0.70, 0.75, 0.80, 0.85, 0.90, 0.95, 1.00},
			axis background/.style={fill=white},
			legend style={at={(0.5,0.95)},legend cell align=left, align=left, draw=white!15!black}
			]
			\addplot[color=violet, line width=0.9pt, only marks] table[col sep=comma] {plots/data_mub2};
			\addlegendentry{Our technique $d=2$}
			\addplot[color=violet, line width=0.9pt] table[col sep=comma] {plots/analytic_mub2};
			\addlegendentry{Exact result $d=2$}

			\addplot[smooth, color=teal, line width=0.9pt, only marks] table[col sep=comma] {plots/data_mub4};
			\addlegendentry{Our technique $d=4$}
			\addplot[color=teal, line width=0.9pt] table[col sep=comma] {plots/analytic_mub4};
			\addlegendentry{Exact result $d=4$}

			\addplot[color=blue, line width=0.9pt, only marks] table[col sep=comma] {plots/data_mub6};
			\addlegendentry{Our technique $d=6$}
			\addplot[color=blue, line width=0.9pt] table[col sep=comma] {plots/analytic_mub6};
			\addlegendentry{Upper bound $d=6$}

			\addplot[color=red, line width=0.9pt, only marks] table[col sep=comma] {plots/data_mub8};
			\addlegendentry{Our technique $d=8$}
			\addplot[color=red, line width=0.9pt] table[col sep=comma] {plots/analytic_mub8};
			\addlegendentry{Exact result $d=8$}

		\end{axis}
	\end{tikzpicture}%
	\caption{Asymptotic key rate in bits per round versus isotropic state visibility $v$, using the MUB protocol from Section \ref{sec:mubprotocol}. For $d=6$ we used a set of bases that is only approximately unbiased. Note that for $d=2,4,8$ our technique closely matches the exact result, while for $d=6$ it is slightly below the upper bound \eqref{eq:tomographic}. Calculations performed with $m=8$.}
	
	\label{fig:mub}
\end{figure}

\subsection{Subspaces protocol}\label{sec:subspace}

When the experimental setup suffers from a high amount of noise, it is advantageous to do a filtration step before running the QKD protocol. That's the idea behind the subspace protocol proposed in Ref.~\cite{doda2021}. In it the Hilbert space of dimension $d$ is partitioned in $d/k$ subspaces of dimension\footnote{The use of equal-sized subspaces is only for simplicity, any partition of $d$ can be used.} $k$. Alice and Bob first check whether they obtained an outcome belonging to the same subspace of the Hilbert space; if they haven't, they discard the round. Otherwise, they proceed with the protocol, with the state conditioned on belonging to the subspace they are in, which is in general less noisy than the state they started with.

The main weakness of the original security analysis is that it approximates the von Neumann entropy by the min-entropy, which leads to a loose lower bound on the secret key rate. Here we can fix this problem. We further note that the subspace technique is very flexible with respect to which protocol is used in the subspaces. Here we use the $d+1$ MUBs protocol from section \ref{sec:mubprotocol}.

For the isotropic state and when $k$ is a prime or prime power the key rate is given by 
\begin{equation}
K_{\text{iso},\text{sub}}(v,k,d) =  p K_\text{iso}(v/p,k),
\end{equation}
where $p = v + \frac{k}{d}(1-v)$ is the probability that Alice and Bob obtain outcomes in the same subspace, and $K_\text{iso}$ is given by Equation \eqref{eq:tomographic}.

We ran that SDP \eqref{eq:qkd_sdp} for $d=8$ and $k=2,4,8$. Note that the MUB protocol is ran in each subspace, so $A^a_0$ and $E_k$ are simply those of the MUB protocol, but with the size of the subspace $k$ playing the role of the dimension $d$. The measurement that is implemented physically is the direct sum of the measurements in the subspaces, but that does not appear explicitly in the key rate calculation. The results are shown in Figure \ref{fig:subspace}. One can see that the numerical results closely match the analytical formula. For the sake of comparison we also showed the much inferior key rates that are obtained when using the min-entropy technique.

\begin{figure}[htb]
	\centering
	\begin{tikzpicture}
		
		\begin{axis}[%
			scale only axis,
			xmin=0.4,
			xmax=1,
			ymin=0,
			ymax=3,
			grid=major,
			xlabel={$v$},
			ylabel={bits},
			xtick={0.40, 0.45, 0.50, 0.55, 0.60, 0.65, 0.70, 0.75, 0.80, 0.85, 0.90, 0.95, 1.00},
			axis background/.style={fill=white},
			legend style={at={(0.55,0.95)},legend cell align=left, align=left, draw=white!15!black}
			]
			\addplot[color=violet, line width=0.9pt, only marks] table[col sep=comma] {plots/data_subspace28};
			\addlegendentry{Our technique $k=2$}
			\addplot[color=violet, line width=0.9pt] table[col sep=comma] {plots/analytic_subspace28};
			\addlegendentry{Exact result $k=2$}
			\addplot[color=violet, line width=0.9pt, dashed] table[col sep=comma] {plots/min_analytic_subspace28};
			\addlegendentry{Min-entropy $k=2$}

			\addplot[smooth, color=teal, line width=0.9pt, only marks] table[col sep=comma] {plots/data_subspace48};
			\addlegendentry{Our technique $k=4$}
			\addplot[color=teal, line width=0.9pt] table[col sep=comma] {plots/analytic_subspace48};
			\addlegendentry{Exact result $k=4$}
			\addplot[color=teal, line width=0.9pt, dashed] table[col sep=comma] {plots/min_analytic_subspace48};
			\addlegendentry{Min-entropy $k=4$}

			\addplot[color=red, line width=0.9pt, only marks] table[col sep=comma] {plots/data_mub8};
			\addlegendentry{Our technique $k=8$}
			\addplot[color=red, line width=0.9pt] table[col sep=comma] {plots/analytic_subspace88};
			\addlegendentry{Exact result $k=8$}
			\addplot[color=red, line width=0.9pt, dashed] table[col sep=comma] {plots/min_analytic_subspace88};
			\addlegendentry{Min-entropy $k=8$}

		\end{axis}
	\end{tikzpicture}%
	\caption{Asymptotic key rate in bits per round versus isotropic state visibility $v$, using the subspace protocol from Section \ref{sec:subspace}. Dashed lines show the much inferior key rates that are obtained with the min-entropy technique. Calculations performed with $m=8$.}
	
	\label{fig:subspace}
\end{figure}
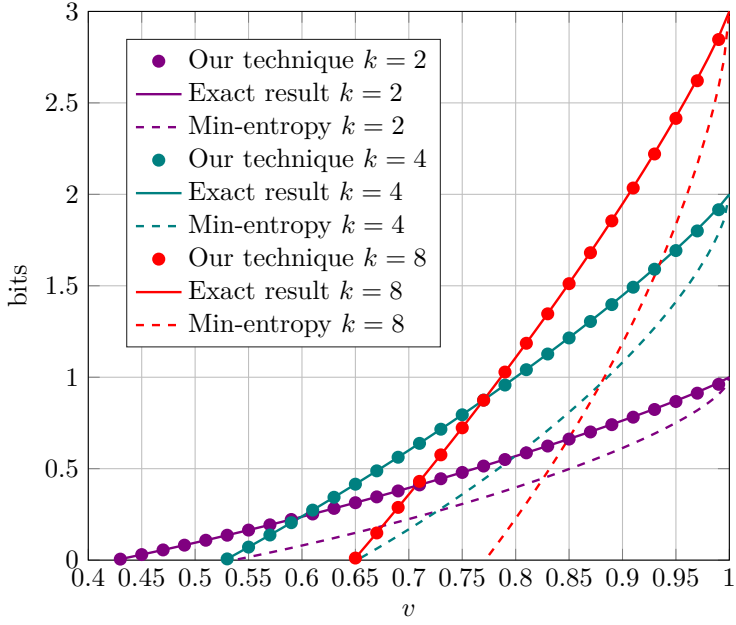

\subsection{Overlapping bases protocol}\label{sec:overlap}

For some experimental setups it is not feasible to measure high-dimensional MUBs, so we consider a simpler protocol where only superpositions of nearest neighbours have to be measured. Alice and Bob use the computational basis as the key basis, $A^a_0 = \proj{a}$, and measure the POVM elements
\begin{equation}
E_{k,i,j} = \proj{v_k^i} \otimes \Proj{v_k^j}^T,
\end{equation}
where $d$ must be even, $k$ goes from 0 to 4, $i,j$ go from 0 to $d-1$, and the vectors $\ket{v_k^i}$ are given by
\begin{subequations}
\begin{gather}
\{\ket{v_0^l}\}_{l=0}^{d-1} = \{\ket{0},\ldots,\ket{d-1}\} \\
\{\ket{v_1^l}\}_{l=0}^{d-1} = \DE{\frac{\ket{0} + \ket{1}}{\sqrt2}, \frac{\ket{0} - \ket{1}}{\sqrt2}, \ldots, \frac{\ket{d-2} + \ket{d-1}}{\sqrt2},\frac{\ket{d-2} - \ket{d-1}}{\sqrt2}} \\
\{\ket{v_2^l}\}_{l=0}^{d-1} = \DE{\frac{\ket{0} + i\ket{1}}{\sqrt2}, \frac{\ket{0} - i\ket{1}}{\sqrt2}, \ldots, \frac{\ket{d-2} + i\ket{d-1}}{\sqrt2},\frac{\ket{d-2} - i\ket{d-1}}{\sqrt2}} \\
\{\ket{v_3^l}\}_{l=0}^{d-1} = \DE{\ket{0},\frac{\ket{1} + \ket{2}}{\sqrt2}, \frac{\ket{1} - \ket{2}}{\sqrt2}, \ldots, \frac{\ket{d-3} + \ket{d-2}}{\sqrt2},\frac{\ket{d-3} - \ket{d-2}}{\sqrt2}, \ket{d-1}} \\
\{\ket{v_4^l}\}_{l=0}^{d-1} = \DE{\ket{0},\frac{\ket{1} + i\ket{2}}{\sqrt2}, \frac{\ket{1} - i\ket{2}}{\sqrt2}, \ldots, \frac{\ket{d-3} + i\ket{d-2}}{\sqrt2},\frac{\ket{d-3} - i\ket{d-2}}{\sqrt2}, \ket{d-1}}.
\end{gather}
\end{subequations}

Note that here we are not measuring only the probabilities that the outcomes are equal or differ by some constant, but instead are taking into account the full experimental data. Several of the projectors we are measuring are linearly dependent, however, so one should be careful to select a linearly independent subset for the data analysis.

This protocol is specially appropriate for energy-time entanglement setups, in which the time of arrival of photons are collected into time-bins. In such setups it is infeasible to measure MUBs, but the superposition of neighbouring bins is accessible with the use of Franson interferometer (see Refs.~\cite{tiranov2017,ecker2019,bulla2022} for experiments using such setups).

We ran the SDP \eqref{eq:qkd_sdp} for $d=4,6,8$. The results are shown in Figure \ref{fig:overlap}. We compare them with key rates obtained with the subspace protocol from Section \ref{sec:subspace}, as that protocol also needs only superpositions between nearest neighbours. We see that for low to moderate amounts of noise the overlap protocol outperforms it.

\begin{figure}[htb]
	\centering
	\begin{tikzpicture}
		
		\begin{axis}[%
			scale only axis,
			xmin=0.7,
			xmax=1,
			ymin=0,
			ymax=3,
			grid=major,
			xlabel={$v$},
			ylabel={bits},
			xtick={0.70, 0.75, 0.80, 0.85, 0.90, 0.95, 1.00},
			axis background/.style={fill=white},
			legend style={at={(0.6,0.95)},legend cell align=left, align=left, draw=white!15!black}
			]
			\addplot[color=teal, line width=0.9pt] table[col sep=comma] {plots/data_overlap4};
			\addlegendentry{Overlap $d=4$}
			\addplot[color=blue, line width=0.9pt] table[col sep=comma] {plots/data_overlap6};
			\addlegendentry{Overlap $d=6$}
			\addplot[color=red, line width=0.9pt] table[col sep=comma] {plots/data_overlap8};
			\addlegendentry{Overlap $d=8$}
			
			\addplot[color=teal, line width=0.9pt, dashed] table[col sep=comma] {plots/analytic_subspace24};
			\addlegendentry{Subspace $k=2,d=4$}
			\addplot[color=blue, line width=0.9pt, dashed] table[col sep=comma] {plots/analytic_subspace26};
			\addlegendentry{Subspace $k=2,d=6$}
			\addplot[color=red, line width=0.9pt, dashed] table[col sep=comma] {plots/analytic_subspace28};
			\addlegendentry{Subspace $k=2,d=8$}
			
		\end{axis}
	\end{tikzpicture}%
	\caption{Asymptotic key rate in bits per round versus isotropic state visibility $v$, using the overlap protocol from Section \ref{sec:overlap}. Dashed line shows the key rate obtained with the subspace protocol from Section \ref{sec:subspace} for the sake of comparison, as that protocol also needs only superpositions between nearest neighbours. We see that for low to moderate amounts of noise the overlap protocol outperforms it. Calculations performed with $m=8$.}
	
	\label{fig:overlap}
\end{figure}
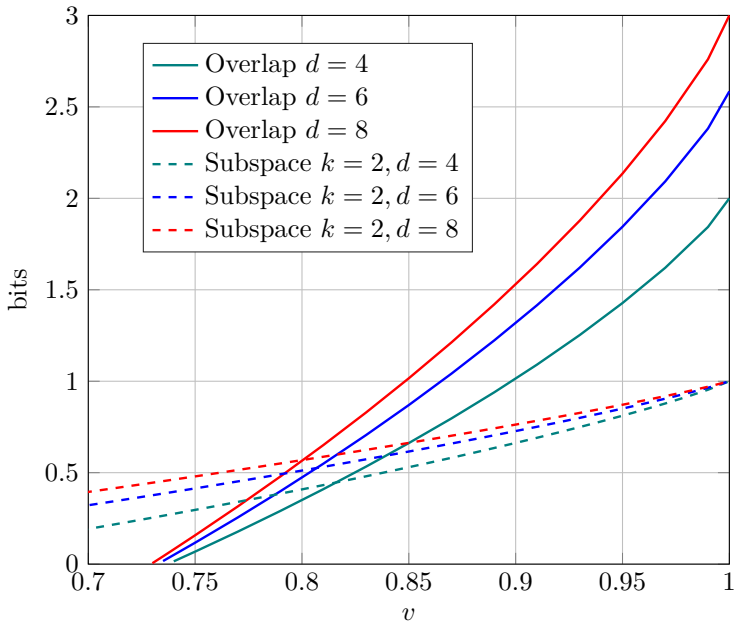  

\section{Dealing with experimental data}\label{sec:experimental}

In order to calculate the conditional entropy with SDP \eqref{eq:qkd_sdp}, we need the exact probabilities of each measurement outcome. Those are not available in real experiments, as it's fundamentally impossible to measure probabilities, even in the absence of error. Instead, one measures relative frequencies, from which one deduces that the probabilities are within some region with some level of confidence. Note that a naïve identification of relative frequencies with probabilities typically leads to probabilities that are not compatible with any quantum state (commonly referred to as quantum state with negative eigenvalues). Computing instead the confidence region compatible with your relative frequencies automatically deals with this issue\footnote{Maximum likelihood estimation \cite{Hradil1997} always produces valid quantum states, but it doesn't provide confidence regions, and as such it is unsuitable for our purposes.}. Accordingly, we need to modify SDP \eqref{eq:qkd_sdp} to not use exact probabilities, but instead minimise the conditional entropy over the confidence region. 

This leads us to the thorny issue of how to compute the confidence region in the first place. The method we choose is Bayesian parameter estimation \cite{buzek1998,schack2001,blumekohout2010}, as it naturally provides a confidence region in the form of the high-density posterior. An exact computation is only feasible for small amounts of data in small dimensions \cite{blumekohout2012,shang2013}, and provides a confidence region with a very complex description, so we have to use approximate methods. The only method we found in the literature is based on using particle filters \cite{ferrie2014,granade2016}. Its complexity scales exponentially with the dimension and is therefore unsuitable for our purposes. We propose thus a new method.

The method we use is based on approximating the likelihood function by a Gaussian. This will be a good approximation whenever the trials are independent and the amount of data is large. It results in a confidence region consisting of the intersection of an ellipsoid with the set of probabilities that can be produced by quantum states. The crucial advantage of this confidence region is that it can be described by semidefinite constraints, and as such be incorporated in our SDP. More formally, the confidence region is given by
\begin{equation}
C = \{\bm{p}\ |\ \left\langle \bm{p}-\bm{f}, \Sigma^{-1}(\bm{p}-\bm{f})\right\rangle \le \chi^2,\quad \exists \sigma\ \tr(\bm{E}\sigma) = \bm{p},\ \tr(\sigma) = 1,\ \sigma \ge 0\},
\end{equation}
where $\bm{f}$ are the measured frequencies, $\bm{p}$ are the probabilities, $\Sigma$ is the covariance matrix of the Gaussian, $\chi$ is a parameter that determines the size of the ellipsoid, and $\bm{E}$ is a vector whose components are the POVM elements. The algorithm to determine $\Sigma$ and $\chi$ is described in Appendix \ref{sec:montecarlo}.

The modified SDP is then given by
\begin{equation}\label{eq:sdp_region}
\begin{gathered}
\min_{\sigma,\bm{p},\{\zeta^a_i,\eta^a_i,\theta^a_i\}_{a,i}} c_m + \sum_{i=1}^m\sum_{a=0}^{n-1} \frac{w_i}{t_i \log 2} \tr\De{(A^a_0 \otimes \id_B)\de{\zeta^a_i + {\zeta^a_i}^\dagger + (1-t_i)\eta^a_i} + t_i\theta^a_i} \\
\text{s.t.}\quad \tr\de{\sigma} = 1,\quad \tr(\bm{E} \sigma) = \bm{p},\quad \left\langle \bm{p}-\bm{f}, \Sigma^{-1}(\bm{p}-\bm{f})\right\rangle \le \chi^2 \\
\forall a,i \quad \Gamma^1_{a,i} := \begin{pmatrix} 
\sigma & \zeta^a_i \\
{\zeta^a_i}^\dagger & \eta^a_i
\end{pmatrix}  \ge 0,\quad \Gamma^2_{a,i} := \begin{pmatrix} 
\sigma & {\zeta^a_i}^\dagger \\
\zeta^a_i & \theta^a_i
\end{pmatrix} \ge 0.
\end{gathered}
\end{equation}

\subsection{Examples}

As an example, we calculated the key rate for the overlap protocol with the data obtained in an experiment reported in Refs.~\cite{bulla2022,bulla2023}. In this experiment time-bin entanglement is produced and distributed over a 10.2km free-space channel over Vienna. The discretisation used here for the key rate calculation was used in Ref.~\cite{bulla2023} to certify genuinely high-dimensional entanglement. We used the data for $d=4$, and obtained $0.4038$. The measurements used were a linearly independent subset of the real projectors $E_{0,i,j}, E_{1,i,j}, E_{3,i,j}$. For comparison, we also computed the key rate with the overlap protocol with $d=4$ and $k=2$, also using only the real projectors, and obtained $0.3868$.

As another example, we calculated the key rate for the MUB protocol with the data obtained by measuring the transverse position-momentum degree-of-freedom of photons using tailored macro-pixel bases, for dimension $d=3$. This data was originally presented in Ref.~\cite{valencia2020} in the context of entanglement certification. Using the full data we got the key rate $1.3310$. For comparison, we also computed the key rate with the original MUB 
protocol from Ref.~\cite{sheridan2010}, that doesn't use the full data, and got $1.3553$. 

Surprisingly, the key rate was \emph{lower} when using the full data. This is because in both cases we are using a flat prior over the parameter space, but the parameter space is much larger when using the full data, which dilutes the effect of Bayesian conditioning. The number of counts in the data set was roughly 600 per setting, which is too small to overcome this effect. If we artificially multiply the number of counts by 10 (which doesn't change the relative frequencies), we get the expected result: key rate $1.4303$ with the full data, and $1.4021$ with the original protocol. This shows that when using the full data we need to make sure we have enough counts to have a good estimate of all the parameters.

\section{Conclusion}

We have developed an SDP hierarchy for calculating asymptotic key rates in quantum key distribution with characterised devices. The algorithm is efficient, easy to implement, and straightforward to use for different protocols. We have also shown how to adapt it to minimise the key rate over the confidence region compatible with the measured statistics, which is necessary to handle real experimental data. Our numerical results show that it closely recovers the known analytical key rates.

We would like to highlight the main advantages of our approach to that of \cite{winick2018,hu2022}. First, our method formulates the problem as an SDP, and thus can use standard solvers, while in \cite{winick2018,hu2022} the problem of calculating the key rate is cast as a general convex optimization problem, which requires custom-programmed algorithms. Also, we believe our method is easier to set up and modify for an arbitrary protocol -- the user only needs to be able to formulate the measurements used for the key generation and the conditional entropy estimation. On the other hand, the methods of \cite{winick2018,hu2022} require deeper understanding of the underlying method, including the definitions of certain linear maps needed to properly define the cost function of each considered protocol.

The fact our method casts the problem as an SDP allows one to obtain an analytical lower bound to the key rate, if desired. One solves the dual problem, approximates the numerical solution by rational numbers, and perturbs the rational solution to make it satisfy the SDP constraints exactly. This is done for example in Appendix C of \cite{Bavaresco2021}. In this way one gets around the finite tolerances of numerical SDP solvers and possible floating-point errors.

Directions for future research include adapting it to protocols that use different security assumptions, that were developed to overcome limitations of standard QKD. Examples are measurement-device-independent QKD \cite{Lo2012}, twin-field QKD \cite{Lucamarini2018}, and using decoy-states \cite{Hwang2003}.

Another important step is to compute the rate for finite keys, which is necessary for real implementations. This can be done using the Entropy Accumulation Theorem \cite{Dupuis2020} as done in \cite{George2022}.

\section{Code availability}

All the code and data necessary to reproduce the results of this paper are available in \url{https://github.com/araujoms/qkd}.

\section{Acknowledgements}
We would like to thank Valerio Scarani for the discussions that triggered this work, and Peter Brown and David Trillo for further valuable discussions, and Hayata Yamasaki for help with the calculations. Hayata Yamasaki acknowledges funding from JST PRESTO Grant Number JPMJPR201A. M.H. would like to acknowledge funding from the European Research Council (Consolidator grant `Cocoquest' 101043705). M.P. acknowledges funding and support from GAMU project MUNI/G/1596/2019 and VEGA Project No. 2/0136/19. M.A. and M.N acknowledge funding from the FWF stand-alone project P 35509-N. A.T acknowledges financial support from the Wenner-Gren foundation. 

\printbibliography
\appendix
\section{Strict feasibility and facial reduction}\label{sec:facialreduction}

\begin{theorem}
The SDP \eqref{eq:qkd_sdp} is strictly feasible if and only if there exists a full rank state $\sigma$ such that $\tr(F_k \sigma) = f_k$ for all $k$.
\end{theorem}
\begin{proof}
To see the if direction, assume that such a state exists. Then let $\zeta^a_i = 0$, and $\eta^a_i = \theta^a_i = \id$. The eigenvalues of $\Gamma^1_{a,i}$ and $\Gamma^2_{a,i}$ are then the eigenvalues of $\sigma$, which are by assumption strictly positive, together with the eigenvalues of $\id$.

To see the converse direction, let $\ket{v}$ be a eigenvector of $\sigma$ with eigenvalue 0. Then since $\Gamma^1_{a,i} = \proj{0}\otimes \sigma + \ketbra{0}{1} \otimes \zeta^a_i + \ketbra{1}{0} \otimes {\zeta^a_i}^\dagger + \proj{1} \otimes \eta^a_i$ we have that $\bra{0}\bra{v} \Gamma^1_{a,i} \ket{0}\ket{v} = 0$. Now if $\Gamma^1_{a,i} \ge 0$, it can be written as $G^\dagger G$ for some matrix $G$. Since $\bra{0}\bra{v} G^\dagger G \ket{0}\ket{v} = \norm{G \ket{0}\ket{v}}_2^2$, we have that $G \ket{0}\ket{v} = 0$, and therefore $\Gamma^1_{a,i} \ket{0}\ket{v} = 0$.
\end{proof}

In order to find out whether a full rank $\sigma$ compatible with the measurement results exist it is enough to maximise the minimum eigenvalue of $\sigma$, which is a simple SDP:
\begin{equation}
\begin{gathered}
\max_{\lambda,\sigma} \lambda \\
\text{s.t.}\quad \tr\de{\sigma} = 1, \quad \forall k\  \tr(F_k \sigma) = f_k, \quad \sigma -\lambda\id \ge 0.
\end{gathered}
\end{equation}
Note that this SDP is always strictly feasible (if a $\sigma$ compatible with the measurement results exists in the first place), because one can set $\lambda = -1$ and then all eigenvalues will be at least 1.

Suppose that $\sigma$ cannot be made full rank. Note that because of the constraints $\Gamma^1_{a,i} \ge 0$ and $\Gamma^2_{a,i}\ge 0$ the matrices $\zeta^a_i$ must have the same support as $\sigma$. Furthermore, the matrices $\theta^a_i$ and $\eta^a_i$ can be chosen to have the same support, as non-zero elements outside it can never decrease the value of the objective. Therefore, to make the SDP strictly feasible we can simply reformulate it explicitly inside the support of $\sigma$.

Let then $\{\ket{v_i}\}_{i=0}^{k-1}$ be an orthonormal basis for the support of $\sigma$, and $V = (\ket{v_0}, \ldots,\ket{v_{k-1}})$ an isometry. We have that $\sigma = V\omega V^\dagger$ for a $k\times k$ quantum state $\omega$, and similarly $\zeta^a_i = V\xi^a_i V^\dagger, \eta^a_i = V\mu^a_i V^\dagger$, and $\theta^a_i = V\nu^a_i V^\dagger$ for  $k\times k$ complex matrices $\xi^a_i,\mu^a_i,\nu^a_i$. Writing then SDP \eqref{eq:qkd_sdp} in terms of the new variables, we have
\begin{equation}\label{eq:reduced_sdp}
\begin{gathered}
\min_{\omega,\{\xi^a_i,\mu^a_i,\nu^a_i\}_{a,i}} c_m + \sum_{i=1}^m\sum_{a=0}^{n-1} \frac{w_i}{t_i \log 2} \tr\De{V^\dagger(A^a_0 \otimes \id_B)V\de{\xi^a_i + {\xi^a_i}^\dagger + (1-t_i)\mu^a_i} + t_i\nu^a_i} \\
\text{s.t.}\quad \tr\de{\omega} = 1,\quad \forall k\  \tr(V^\dagger F_k V \omega) = f_k\\
\forall a,i \quad \Gamma^1_{a,i} := \begin{pmatrix} 
\omega & \xi^a_i \\
{\xi^a_i}^\dagger & \mu^a_i
\end{pmatrix}  \ge 0,\quad \Gamma^2_{a,i} := \begin{pmatrix} 
\omega & {\xi^a_i}^\dagger \\
\xi^a_i & \nu^a_i
\end{pmatrix} \ge 0.
\end{gathered}
\end{equation}

\subsection{Example}

Suppose Alice and Bob share a state of local dimension 2, and are estimating the probability of getting equal outcomes in the $Z$ and $X$ bases. That is
\begin{gather}
F_0 = \proj{0} + \proj{1}\\
F_1 = \proj{+} + \proj{-}.
\end{gather}
Suppose then that the probabilities they estimate\footnote{This is of course not a realistic scenario. Even in the very unlikely case where they always get equal outcomes in the $Z$ basis, they would be fools to estimate the probability to be 1.} are $f_0 = 1$ and $f_1 = x$, for some $x \in (0,1)$. Then necessarily we have that $\sigma\ket{01} = \sigma\ket{10} = 0$, so $\sigma$ is supported in the 2-dimensional subspace spanned by $\ket{\phi^+}$ and $\ket{\phi^-}$. The isometry is then $V = (\ket{\phi^+},\ket{\phi^-})$, and with it we can solve the reduced SDP \eqref{eq:reduced_sdp}. It turns out that in this case facial reduction made no practical difference: the numerical solver could also handle this problem without the reduction, the only difference was that the reduced problem was much faster.

Suppose now that $x=1$, so the only possible solution is $\sigma = \proj{\phi^+}$, and the isometry is $V = \ket{\phi^+}$. In this case the solver cannot handle the original problem, it suffers from numerical instabilities and gives out an incorrect answer. The reduced problem was solved correctly without difficulties.

\section{Symmetrisation}\label{sec:symmetrisation}

In this appendix we explore in more detail the symmetrisation techniques mentioned in Section \ref{sec:mainsdp}.

\subsection{Complex conjugation}\label{sec:symmetry_real}

The simplest symmetrisation applies when $A^a_0$ and $E_k$ are real matrices for all $a,k$. Then the transformation
\begin{subequations}\label{eq:symmetry_real}
\begin{gather}
\sigma \mapsto \sigma^*, \\
\zeta^a_i \mapsto {\zeta^a_i}^*, \\
\eta^a_i \mapsto {\eta^a_i}^*,\\
\theta^a_i \mapsto {\theta^a_i}^*,
\end{gather}
\end{subequations}
will leave the objective invariant, as
\begin{multline}
c_m + \sum_{i=1}^m\sum_{a=0}^{n-1} \frac{w_i}{t_i \log 2} \tr\De{(A^a_0 \otimes \id_B)\de{{\zeta^a_i}^* + {{\zeta^a_i}^*}^\dagger + (1-t_i){\eta^a_i}^*} + t_i{\theta^a_i}^*} \\
= c_m + \sum_{i=1}^m\sum_{a=0}^{n-1} \frac{w_i}{t_i \log 2} \tr\De{({A^a_0}^T \otimes \id_B)\de{\zeta^a_i + {\zeta^a_i}^\dagger + (1-t_i)\eta^a_i} + t_i\theta^a_i},
\end{multline}
and map satisfied constraints to satisfied constraints, as $\tr(\sigma^*) = \tr(\sigma)$, $\tr(E_k \sigma^*) = \tr(E_k^T \sigma)$, and ${\Gamma^g_{a,i}}^* \ge 0$ iff ${\Gamma^g_{a,i}} \ge 0$.

This implies that whenever $\sigma,\{\zeta^a_i,\eta^a_i,\theta^a_i\}_{a,i}$ is feasible solution to the SDP, $\sigma^*,\{{\zeta^a_i}^*,{\eta^a_i}^*,{\theta^a_i}^*\}_{a,i}$ will also be a feasible solution with the same value, and the same will be true for any convex combination of these two sets. Using then equally weighted convex combination, we see that we can choose all of them to be real.

Even if the $E_k$ are not real, this symmetrisation also applies if for all $k$ there exists a $k'$ such that $E_k^T = E_{k'}$ and $f_k = f_{k'}$. This holds for example for the MUB protocol from Section \ref{sec:mubprotocol} with the probabilities of the isotropic state. It will certainly not hold for real experimental data (or even simulated data), so we refrained from using this in our calculations.

\subsection{Permutation}\label{sec:permutation}

Assume now that $A^a_0 = \proj{a}$, that is, we are using the computational basis as the key basis. Then for any permutation $\pi$ and any unitary $V$ the objective will be invariant under the mapping
\begin{subequations}\label{eq:symmetry}
\begin{gather}
\zeta^a_i \mapsto (U_\pi^\dagger \otimes V^\dagger) \zeta^{\pi(a)}_i (U_\pi \otimes V), \\
\eta^a_i \mapsto (U_\pi^\dagger \otimes V^\dagger) \eta^{\pi(a)}_i (U_\pi \otimes V), \\
\theta^a_i \mapsto (U_\pi^\dagger \otimes V^\dagger) \theta^{\pi(a)}_i (U_\pi \otimes V), \\
\sigma \mapsto (U_\pi^\dagger \otimes V^\dagger) \sigma (U_\pi \otimes V),
\end{gather}
\end{subequations}
where $U_\pi$ is the permutation matrix such that $U_\pi\ket{a} = \ket{\pi(a)}$ for all $a$.

This mapping also doesn't affect the satisfiability of the constraints of the SDP, except for the constraint $\tr(E_k \sigma) = f_k$, which is mapped onto $\tr\De{E_k (U_\pi^\dagger \otimes V^\dagger) \sigma (U_\pi \otimes V)} = f_k$. Therefore, the mapping \eqref{eq:symmetry} is a symmetry of the SDP if
\begin{equation}\label{eq:measurement_symmetry}
E_k = (U_\pi \otimes V) E_k (U_\pi^\dagger \otimes V^\dagger) \quad \forall k.
\end{equation}
On the flip side, whenever our measurement bases satisfy such a permutation symmetry, we can use that symmetry to simplify the SDP.

Assuming that this holds for some $\pi$ and $V$, let $S_\pi$ be a set of representatives for the orbits of $\pi$, and $K_\pi(a)$ the size of the orbit of representative $a$. For example, if $\pi$ is such that $\pi(0) = 2, \pi(1)=1$, and $\pi(2) = 0$, then $S_\pi = \{0,1\}, K_\pi(0) = 2$, and $K_\pi(1) = 1$. Then SDP \eqref{eq:qkd_sdp} simplifies to 
\begin{equation}\label{eq:qkd_sdp_symmetrized}
\begin{gathered}
\min_{\sigma,\{\zeta^a_i,\eta^a_i,\theta^a_i\}_{a,i}} c_m + \sum_{i=1}^m\sum_{a \in S_\pi} K_\pi(a) \frac{w_i}{t_i \log 2} \tr\De{(A^a_0 \otimes \id_B)\de{\zeta^a_i + {\zeta^a_i}^\dagger + (1-t_i)\eta^a_i} + t_i\theta^a_i} \\
\text{s.t.}\quad \tr\de{\sigma} = 1,\quad \sigma = (U_\pi^\dagger \otimes V^\dagger) \sigma (U_\pi \otimes V), \quad \forall k\  \tr(E_k \sigma) = f_k\\
\forall a \in S_\pi,\forall i \quad \Gamma^1_{a,i} := \begin{pmatrix} 
\sigma & \zeta^a_i \\
{\zeta^a_i}^\dagger & \eta^a_i
\end{pmatrix}  \ge 0,\quad \Gamma^2_{a,i} := \begin{pmatrix} 
\sigma & {\zeta^a_i}^\dagger \\
\zeta^a_i & \theta^a_i
\end{pmatrix} \ge 0.
\end{gathered}
\end{equation}

This symmetrisation applies for example to the original MUB protocol from Ref.~\cite{sheridan2010}, without generalisations we introduced in Section \ref{sec:mubprotocol}. In the original protocol they measure only the projectors
\begin{equation}
E_{k,l} := \sum_{i=0}^{d-1} \Pi^i_k \otimes {\Pi^{i\oplus l}_k}^T
\end{equation}
for prime $d$, with $l$ going from 0 to $d-2$.

In this case the MUBs are given by the eigenvectors of the Weyl operators, and it is easy to see that Equation \eqref{eq:measurement_symmetry} is satisfied with $V=U_\pi$ for the cyclic permutation $\pi(a) = a\oplus 1$, where addition is modulo $d$. Then $S_\pi = \{0\}$, which has an orbit of size $d$. This provides a dramatic increase in performance, cutting running time and memory usage by roughly a factor of $d$.

This symmetrisation also applies for the overlap protocol from Section \ref{sec:mubprotocol} if we modify it to measure only the probabilities that the outcomes are equal. Then Equation \eqref{eq:measurement_symmetry} is satisfied with $V=U_\pi$ for the permutation $\pi(a) = d-1-a$. Then $S_\pi = \{0, \ldots, d/2-1\}$, and all orbits have size two.

\subsection{Lower bound}

This is not a symmetrisation technique, but it is worth noting that one can reduce the memory usage of SDP \eqref{eq:qkd_sdp} by using the following lower bound, analogous to the one used in Ref.~\cite{brown2021}:
\begin{multline}\label{eq:lowerbound}
\min_{\sigma,\{\zeta^a_i,\eta^a_i,\theta^a_i\}_{a,i}} c_m + \sum_{i=1}^m\sum_{a=0}^{n-1} \frac{w_i}{t_i \log 2} \tr\De{(A^a_0 \otimes \id_B)\de{\zeta^a_i + {\zeta^a_i}^\dagger + (1-t_i)\eta^a_i} + t_i\theta^a_i} \\
\ge c_m + \sum_{i=1}^m \frac{w_i}{t_i \log 2} \min_{\sigma,\{\zeta^a_i,\eta^a_i,\theta^a_i\}_{a}} \sum_{a=0}^{n-1} \tr\De{(A^a_0 \otimes \id_B)\de{\zeta^a_i + {\zeta^a_i}^\dagger + (1-t_i)\eta^a_i} + t_i\theta^a_i}.
\end{multline}
While the original SDP requires $m n$ variables $\{\zeta^a_i, \eta^a_i, \theta^a_i\}_{a,i}$, the lower bound breaks it down into $m$ different SDPs with only $n$ variables $\{\zeta^a_i, \eta^a_i, \theta^a_i\}_a$ each. The lower bound is in general not tight, and the running time of the $m$ SDPs is often longer than the running time of the original SDP, so it should only be used on systems with insufficient memory. We haven't used this lower bound in any of the calculations in this paper.

\section{Determining the confidence region}\label{sec:montecarlo}

The idea behind Bayesian parameter estimation is simple: we wish to estimate some parameters $\bm{\theta}$ based on some data $D$, given a prior distribution over the parameters $p(\bm{\theta})$. The posterior distribution of the parameters is therefore given by
\[p(\bm{\theta}|D) = \frac{p(D|\bm{\theta})p(\bm\theta)}{\int \mathrm{d}\bm\theta' p(D|\bm{\theta'})p(\bm\theta')},\]
the resulting estimate for the parameters is the expected value
\begin{equation}\label{eq:expectedvalue}
\hat{\bm\theta} := \int \mathrm{d}\bm\theta p(\bm{\theta}|D)\bm\theta,
\end{equation}
and the $1-\alpha$ credible region (Bayesian version of confidence region) is defined as the high density posterior, that is, the set
\begin{equation}
S_\gamma := \{\bm\theta\; ; \;p(\bm{\theta}|D) \ge \gamma\},
\end{equation}
where $\gamma$ is defined as the supremum of $\gamma'$ such that
\begin{equation}\label{eq:credible_region}
\int_{S_{\gamma'}}\mathrm{d}\bm\theta p(\bm{\theta}|D) \ge 1-\alpha.
\end{equation}
The required integrals are in general very difficult to compute. We propose here to compute them by approximating the likelihood function $p(D|\bm{\theta})$ by a Gaussian, restricting the prior $p(\bm\theta)$ to be either flat or Gaussian, and then applying Monte Carlo methods.

The reason for doing so is that the product of two Gaussians is again a Gaussian, and so the posterior distribution will be the intersection of a Gaussian with the space of valid $\bm\theta$, call it $\Theta$ (in our case the set of probabilities that can be produced by a quantum state). This makes the credible region  easy to describe; if the Gaussian of the posterior has mean $\bm\mu$ and covariance matrix $\Sigma$, that is, if 
\begin{equation}
G(\bm\theta) = \frac1{\sqrt{\det(2\pi\Sigma)}}\exp\de{-\frac12\left\langle \bm\theta-\bm{\mu}, \Sigma^{-1} (\bm\theta-\bm{\mu})\right\rangle},
\end{equation}
then the credible region will be the set of $\bm\theta$ such that $\bm\theta \in \Theta$ and
\begin{equation}\label{eq:credible_region_chi}
\left\langle \bm\theta-\bm{\mu}, \Sigma^{-1} (\bm\theta-\bm{\mu})\right\rangle \le \chi^2
\end{equation}
for some $\chi$. The main advantage of doing so is that inequality \eqref{eq:credible_region_chi} is an SDP constraint, so if the condition $\bm\theta \in \Theta$ also is, we can bound quantities of interest maximising or minimising them over the credible region.

To compute the credible region we then start with a initial guess $\chi_0 = Q^{-1}(\kappa/2,0,1-\alpha)$, where $Q$ is the regularised incomplete gamma function and $\kappa$ the number of parameters we are estimating. We then compute integral \eqref{eq:credible_region} via Monte Carlo: we sample $n$ times a $\bm\theta$ from $p(\bm{\theta}|D)$ and count how many times $k$ it respects inequality \eqref{eq:credible_region_chi}. Our estimate of the integral is then $k/n$. We then increase or decrease $\chi_0$ until we our estimate is close enough to $1-\alpha$, via binary search.

To sample $\bm\theta$ from $p(\bm{\theta}|D)$, there are two methods: the easiest one is rejection sampling: sample $\bm\theta$ from $G(\bm\theta)$, which is very easy to do, and check whether $\bm\theta \in \Theta$. If yes, accept the sample, otherwise reject. It might be the case that the probability of getting a sample $\bm\theta \in \Theta$ is too low for this method to be viable, usually when $\bm\mu \not\in \Theta$. In this case we need to sample $\bm\theta$ with the Metropolis algorithm.

Up to this point, the discussion of Bayesian parameter estimation has been quite generic. In order to apply it to our problem we need to specify the likelihood function and the Gaussian that approximates it. 

If one assumes the random variables are independent and identically distributed (i.i.d.), the likelihood function is given simply by the multinomial distribution. Let's say one performed a set of $n$ measurements with $k_i$ outcomes each, obtaining counts $\mathbf{N} := \{N^w_i\}_{w=0,i=0}^{k_i-1,n-1}$ with probabilities $\mathbf{p}:=\{p(w|i)\}_{w=0,i=0}^{k_i-1,n-1}$. The likelihood function is then
\begin{equation}
L(\mathbf{p}) := \exp \langle \mathbf{N},\log(\mathbf{p})\rangle,
\end{equation}
where $\langle \cdot,\cdot \rangle$ denotes the inner product. We approximate it by the Gaussian with multinomial mean and covariances.

The mean $\bm{\mu} := \{\mu(w|i)\}_{w=0,i=0}^{k_i-2,n-1}$ is given by
\begin{equation}
\mu(w|i) = \frac{N^w_i}{N_i},
\end{equation}
where we are omitting the relative frequency of the last outcome in order to avoid degeneracy, and $N_i := \sum_{w'=0}^{k_i-1}N^{w'}_i$.

The covariance matrix $\Sigma$ is given by 
\begin{equation}
\Sigma = \bigoplus_{i=0}^{n-1} (\operatorname{diag}(\bm{\mu}_i) - \bm{\mu}_i \bm{\mu}_i^T)/N_i,
\end{equation}
where $\bm{\mu}_i := \{\mu(w|i)\}_{w=0}^{k_i-2}$. If any of the counts $N^w_i$ is equal to zero, this will lead to a singular covariance matrix, which can not be used in the Gaussian. A simple trick to get around this is to change this count to one for the calculation of the covariance matrix. This makes it no longer singular, and still provides a decent approximation to the multinomial distribution, provided the number of counts is high enough for the Gaussian to be a sensible approximation.

Without the i.i.d. assumption, the covariance matrix is no longer a simple function of the mean, but needs to be estimated separately.

\end{document}